\documentclass[10pt,english,conference]{IEEEtran}

\IEEEoverridecommandlockouts

\usepackage{amsthm,amsmath,amssymb}
\usepackage{graphicx}
\usepackage{cite}
\usepackage{kbordermatrix}

\usepackage{tikz}
\usetikzlibrary{arrows,matrix}

\newcommand{\ie}{i.e.}

\newcommand{\eg}{e.g.}

\newcommand{\mc}[1]{\mathcal{#1}}

\newcommand{\mbb}[1]{\mathbb{#1}}

\newcommand{\set}[1]{\mathcal{#1}}
\newcommand{\compl}[1]{\overline{#1}}

\DeclareMathOperator{\ind}{ind}
\DeclareMathOperator{\lind}{lind}
\DeclareMathOperator{\vlind}{\overrightarrow{\lind}}
\DeclareMathOperator{\minrank}{minrank}

\newtheorem{lemma}{Lemma}
\newtheorem{theorem}{Theorem}
\newtheorem{corollary}{Corollary}

\newtheorem{definition}{Definition}

\newtheorem{example}{Example}

\author{\IEEEauthorblockN{Javad B. Ebrahimi\IEEEauthorrefmark{1}, Mahdi~Jafari~Siavoshani\IEEEauthorrefmark{2}}
\IEEEauthorblockA{\IEEEauthorrefmark{1}Institute of Network Coding,
Chinese University of Hong Kong, Hong Kong\\
\IEEEauthorrefmark{2}Computer Engineering Department, Sharif University of Technology, Tehran, Iran\\
Email: \text{javad@inc.cuhk.edu.hk}, \text{mjafari@sharif.edu}}
\thanks{The work described in this paper was partially supported by a grant from University Grants Committee of the Hong Kong Special Administrative Region, China (Project No. AoE/E-02/08).}
}

\title{On Index Coding and Graph Homomorphism}

\begin{document}
\maketitle

\begin{abstract}
In this work, we study the problem of \emph{index coding} from graph homomorphism perspective. We show that the minimum broadcast rate of an index coding problem for different variations of the problem such as non-linear, scalar, and vector index code, can be upper bounded by the minimum broadcast rate of another index coding problem when there exists a homomorphism from the complement of the side information graph of the first problem to that of the second problem. As a result, we show that several upper bounds on scalar and vector index code problem are special cases of one of our main theorems.

For the linear scalar index coding problem, it has been shown in \cite{YosBirkJayKol-IT11} that the binary linear index of a graph is equal to a graph theoretical parameter called \emph{minrank} of the graph. For undirected graphs, in \cite{ChlamHaviv-CoRR11} it is shown that $\minrank(G)=k$ if and only if there exists a homomorphism from $\bar{G}$ to a predefined graph $\bar{G_k}$. Combining these two results, it follows that for undirected graphs, all the digraphs with linear index of at most $k$ coincide with the graphs $G$ for which there exists a homomorphism from $\bar{G}$ to $\bar{G_k}$. In this paper, we give a direct proof to this result that works for digraphs as well. 

We show how to use this classification result to generate lower bounds on scalar and vector index. In particular, we provide a lower bound for the scalar index of a digraph in terms of the chromatic number of its complement. 

Using our framework, we show that by changing the field size, linear index of a digraph can be at most increased by a factor that is independent from the number of the nodes. 
%
\end{abstract}

\section{Introduction}\label{sec:Introduction}
The \emph{index coding} problem, first introduced by Birk and Kol in the context of satellite communication \cite{BirkKol-IndexCoding-INFOCOM98}, has received significant attention during past years (see for example \cite{AlonLubetStavWeinHassid-Focs08, LubetStav-IT09, RouaySprintGeorgh-IT10, BlaKleLub-CoRR10, YosBirkJayKol-IT11, BerLangberg-ISIT11, HavivLangberg-ISIT12, TehraniDimakisNeely-ISIT12, MalCadJafar-CoRR12, ShanmugamDimakisLangberg-CoRR13, ArbabBanderKimSasogluWang-ISIT13}). This problem has many applications such as satellite communication, multimedia distribution over wireless networks, and distributed caching. Despite its simple description, the index coding problem has a rich structure and it has intriguing connections to some of the information theory problems. It has been recently shown that the feasibility of any network coding problem can be reduced to an equivalent feasibility problem in the index coding problem (and vice versa) \cite{EffrosSalimLangberg-CoRR12}. Also an interesting connection between index coding problem and interference alignment technique has been appeared in \cite{MalCadJafar-CoRR12}.

In this work, we focus on the index coding problems that can be represented by a side information graph (defined in \S\ref{sec:ProbStatement}), \ie, user demands are distinct and there is exactly one receiver for each message. For this case we consider the framework for studying the index coding problem that uses ideas from graph homomorphism. More precisely, we show that the minimum broadcast rate of an index coding problem (linear or non-linear) can be upper bounded by the minimum broadcast rate of another index coding problem if there exists a homomorphism from the (directed) complement of the side information graph of the first problem to that of the second problem. Consequently, we show that the chromatic and fractional chromatic number upper bound are special cases of our results (\eg, see \cite{YosBirkJayKol-IT11,BlaKleLub-CoRR10}).

For the case of linear scalar, we also prove the opposite direction, namely, we show that for every positive integer $k$ and prime power $q$, there exits a digraph $H^{q}_k$ such that the $q$-arry linear index of $H^q_k$ is at most $k$ and the complement of any digraph whose $q$-arry linear index is also at most $k$ is homomorphic to $\compl{H^{q}_k}$. The set of graphs $H^{q}_k$ are analogous to the ``graph family $G_k$'' defined in \cite{ChlamHaviv-CoRR11} for studying a parameter of the graph called \emph{minrank}. In contrast to those graphs, $H^{q}_k$ are defined for arbitrary finite fields as opposed to the binary field and more importantly, they can be utilised to study the linear index code even if the graphs of interest are directed. Moreover, our proof does not use the result of \cite{YosBirkJayKol-IT11} about the equivalence between the minrank and the linear index of graphs.

Using the reduction of the scalar index coding problem to the homomorphism problem and the notion of increasing functions on the set of digraphs, we provide a family of lower bounds on the binary index of digraphs. As a particular example of such lower bounds, we extend the the previously known bound $\log_q(\chi(\compl{G})) \le \lind_q(G)$ \cite{LangbergSprinston-ISIT08} from $q=2$ to arbitrary $q$.

As an application of our work, we show a connection between $\lind_p(\cdot)$ and $\lind_q(\cdot)$ when $p$ and $q$ are different prime powers.

The remainder of this paper is organised as follows. In \S\ref{sec:ProbStatement} we introduce notation, some preliminary concepts about graph homomorphism and give the problem statement. The main results of the paper and their proofs are presented in \S\ref{sec:IndCode_via_GrphHom} and \S\ref{sec:EquivFormLinScalBinarIndexCoding}. In \S\ref{sec:Application}, some applications of our main results are stated. 
The omitted proofs can be found in \cite{EbrahimiJafari-TechReport13}.

\section{Notation and Problem Statement}\label{sec:ProbStatement}

\subsection{Notation and Preliminaries}
For convenience, we use $[m:n]$ to denote for the set of natural numbers $\{m,\ldots,n\}$. For any set $\set{A}$, we use $\mc{P}^\star(\set{A})$ to denote for all of the non-empty subsets of $A$. Let $x_1,\ldots,x_n$ be a set of variables. Then for any subset $\set{A}\subseteq [1:n]$ we define $x_{\set{A}}\triangleq (x_i:i\in\set{A})$.

A directed graph (digraph) $G$ is represented by $G(V,E)$ where $V$ is the set of vertices and $E\subseteq (V\times V)$ is the set of edges. For $v\in V(G)$ we denote by $N^+_G(v)$ as the set of outgoing neighbours of $v$, \ie, $N_G^+(v)=\{u\in V:(v,u)\in E(G)\}$. For a digraph $G$ we use $\compl{G}$ to denote for its \emph{directional complement}, \ie, $(u,v)\in E(G)$ iff $(u,v)\notin E(\compl{G})$.

\begin{definition}[Homomorphism, see~\cite{HellNese-Book04-GraphHomomorphism}]
Let $G$ and $H$ be any two digraphs. A \emph{homomorphism} from $G$ to $H$, written
as $\phi:G\mapsto H$ is a mapping $\phi:V(G)\mapsto V(H)$ such that $(\phi(u),\phi(v))\in E(H)$ whenever $(u,v)\in E(G)$.
If there exists a homomorphism of $G$ to $H$ we write $G\rightarrow H$,
and if there is no such homomorphism we shall write $G\nrightarrow H$. In the former case we say that $G$ is homomorphic to $H$.
\end{definition}

\begin{definition}
On the set of all loop-less digraphs $\set{G}$, we define the partial pre
order ``$\preccurlyeq$''  as follows. For every pair of $G,H\in\set{G}$,
$G\preccurlyeq H$ if and only if there exists a homomorphism $\phi:\compl{G}\mapsto \compl{H}$. It is straightforward to see that ``$\preccurlyeq$'' is reflexive and transitive. Moreover, if $G\preccurlyeq H$
and $H\preccurlyeq G$, then the digraphs $\compl{G}$ and $\compl{H}$ are
homomorphically equivalent (\ie, $\compl{G}\rightarrow\compl{H}$ and
$\compl{H}\rightarrow\compl{G}$). In this case we write $G\sim H$.
\end{definition}
Notice that homomorphically equivalence does not imply isomorphism between graphs (digraphs). For example, all the bipartite graphs are homomorphically equivalent to $K_2$ and therefore are homomorphically equivalent to each other but they are not necessarily isomorphic. 
\begin{definition}
Let $\set{D}\subseteq\set{G}$ be an arbitrary set of digraphs. A mapping $h:\set{D}\mapsto\mbb{R}$ is called \emph{increasing} over $\set{D}$ if for every 
$G,H\in\set{D}$ such that $G\preccurlyeq H$ then $h(G)\le h(H)$.
\end{definition}

\subsection{Problem Statement}
Consider the communication problem where a transmitter aims to communicate
a set of $m$ messages $x_1,\ldots,x_m\in\set{X}$ to $m$ receivers by broadcasting $\ell$ 
symbols $y_1,\ldots,y_{\ell}\in\set{Y}$, over a public noiseless channel. We assume that for each $j\in[1:m]$,
the $j$th receiver has access to the side information $x_{\set{A}_j}$, \ie, a subset
$\set{A}_j \subseteq [1:m]\setminus \{j\}$ of messages. Each receiver $j$ intends to recover $x_j$ from $(y^\ell,x_{\set{A}_j})$.

This problem, which is a basic setting of the \emph{index coding} problem, 
can be represented by a \emph{directed side information graph} $G(V,E)$ where $V$ represents the 
set of receivers/messages and there is an edge from node $v_i$ to $v_j$, \ie, 
$(v_i,v_j)\in E$, if the $i$th receiver has packet $x_j$ as side information.
An index coding problem, as defined above, is completely characterized by 
the side information sets $\set{A}_j$. 
%
%

In the following definitions, we formally define validity of an index codes and some other basic concepts in index coding (see also \cite{AlonLubetStavWeinHassid-Focs08,BlaKleLub-CoRR10}, and \cite{ShanmugamDimakisLangberg-CoRR13}).

\begin{definition}[Valid Index Code]
A \emph{valid index code} for $G$ over an alphabet $\set{X}$ is a set 
$(\Phi,\{\Psi_i\}_{i=1}^m)$
consisting of: (i) an encoding function $\Phi:\set{X}^m\mapsto \set{Y}^\ell$
which maps $m$ source messages to a
transmitted sequence of length $\ell$ of symbols from $\set{Y}$;
(ii) a set of $m$ decoding functions $\Psi_i$ such that for each 
$i\in[1:m]$ we have $\Psi_i(\Phi(x_1,\ldots,x_m),x_{\set{A}_i})=x_i$.
\end{definition}

\begin{definition}
Let $G$ be a digraph, and $\set{X}$ and $\set{Y}$ are the source and the message alphabet, respectively.
\\
(i) The ``\emph{broadcast rate}'' of an index code $(\Phi,\{\Psi_i\})$ is defined as $\ind_{\set{X}}(G,\Phi,\{\Psi_i\})\triangleq \frac{\ell\log|\set{Y}|}{\log|\set{X}|}$.\\
(ii) The ``\emph{index}'' of $G$ over $\set{X}$, denoted by $\ind_{\set{X}}(G)$ is defined as
$\ind_{\set{X}}(G)= \inf_{\Phi,\{\Psi_i\} } \ind_{\set{X}}(G,\Phi,\{\Psi_i\})$.\\
(iii) If $\set{X}=\set{Y}=\mbb{F}_q$ (the $q$-element finite field for some prime power $q$), the ``\emph{scalar linear index}'' of $G$, denoted by $\lind_{q}(G)$ is defined as $\lind_{q}(G) \triangleq \inf_{\Phi,\{\Psi_i\}} \ind_{\mbb{F}_q} (G,\Phi,\{\Psi_i\})$ in which the infimum is taken over the coding functions of the form $\Phi=(\Phi_1,\ldots,\Phi_\ell)$ and each $\Phi_i$ is a linear combination of $x_j$'s with coefficients from $\mbb{F}_q$.\\
(iv) If $\set{X}=\mbb{F}_q^t$ and $\set{Y}=\mbb{F}_q$, the \emph{vector linear index} for $G$, denoted by $\vlind_{q^t}(G)$ is defined as $\vlind_{q^t}(G) \triangleq \inf_{\Phi,\{\Psi_i\}} \ind_{\mbb{F}_q^t}(G,\Phi,\{\Psi_i\})$ where the infimum is taken over all coding functions $\Phi=(\Phi_1,\ldots,\Phi_\ell)$ such that $\Phi_i:\mbb{F}_q^{tm}\mapsto\mbb{F}_q$ are $\mbb{F}_q$-linear functions.\\
(v) The ``\emph{minimum broadcast rate}'' of the index coding problem of $G$ is defined as
$\ind(G) \triangleq \inf_{\set{X}}\ \inf_{\Phi,\{\Psi_i\} } \ind_{\set{X}}(G,\Phi,\{\Psi_i\})$.
\end{definition}

\section{Index Coding via Graph Homomorphism}\label{sec:IndCode_via_GrphHom}
In this section, we will explain a method for designing index codes from another instance of index coding problem when there exists a homomorphism from the complement of the side information graph of the first problem to that of the second one. As an application to this result, we will show in \S\ref{sec:Application} that some of the previously known results about index code design are special types of our general method. 

\begin{theorem}\label{thm:MainResult_indRate_Homomorphism}
Consider two instances of the index coding problems over the digraphs $G$ and
$H$ with the source alphabet $\set{X}$. 
If $G\preccurlyeq H$ then
\begin{equation*}
\ind_{\set{X}}(G) \le \ind_{\set{X}}(H).
\end{equation*}
In other words, the function $\ind_{\set{X}}(\cdot)$ is a 
non-decreasing function on the pre order set $(\set{G},\preccurlyeq)$.
\end{theorem}

First we explain the proof idea of 
Theorem~\ref{thm:MainResult_indRate_Homomorphism} which is as follows. If $G\preccurlyeq H$, by definition
there exists a homomorphism $\phi:\compl{G}\mapsto\compl{H}$. Notice that
the function $\phi$ maps the vertices of $\compl{G}$ to the vertices of 
$\compl{H}$. Thus we can also consider, $\phi$ as a function from $V(G)$ to $V(H)$. For every vertex $w\in V(H)$, we denote by $\phi^{-1}(w)$ to be
the set of all the vertices $v\in V(G)$ such that $\phi(v)=w$; (see Figure~\ref{fig:Homomorphism_and_InvImage}).
This way, we partition the vertices of $G$ into the classes of the form $\phi^{-1}(w)$ where $w\in V(H)$.

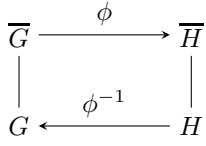
\begin{figure}
\begin{center}
\begin{tikzpicture}
  \matrix (m) [matrix of math nodes,row sep=2em,column sep=5em,minimum width=1em]
  {
     \compl{G} & \compl{H} \\
     G & H \\};
     
  \path[-stealth]
    (m-1-1) edge [-] (m-2-1)
                  edge node [above] {$\phi$} (m-1-2)
    (m-1-2) edge [-] (m-2-2)
    (m-2-2) edge node [above] {$\phi^{-1}$} (m-2-1);     
\end{tikzpicture}
\end{center}
\vspace{-3mm}
\caption{Homomorphism $\phi$ maps the vertices of $G$ to the vertices of $H$. The pre-image $\phi^{-1}$ of the homomorphism can be considered as a mapping from $V(H)$ to $V(G)$, \ie, for every $w\in V(H)$, $\phi^{-1}(w)$ is the set of all the vertices $v$ in $V(G)$ such that $\phi(v)=w$.}\label{fig:Homomorphism_and_InvImage}
\end{figure}

Next, we take an optimal index code for $H$ over the source alphabet $\set{X}$ that achieves the rate $\ind_{\set{X}}(H)$. Then, we show that we can treat every part $\phi^{-1}(w)$ as a single node and translate the index code of $H$ to one for $G$. This shows the statement of the theorem, \ie, $\ind_{\set{X}}(G) \le \ind_{\set{X}}(H)$.

Before we formally define the translation and verify its validity, we will state two technical lemmas that will be required later in the proof of Theorem~\ref{thm:MainResult_indRate_Homomorphism}.

\begin{lemma}\label{lem:Homomorphism_bundle}
(i) For every $w\in V(H)$, $\phi^{-1}(w)$ is a clique in $G$.\\
(ii) If $w_2\in N_H^+(w_1)$, $v_1\in\phi^{-1}(w_1)$, and $v_2\in\phi^{-1}(w_2)$
then $v_2\in N_G^+(v_1)$. (Also see Figure~\ref{fig:Homomorphism_bundle}).
\end{lemma}

\begin{figure}
\begin{center}
\begin{tikzpicture}[->, >=stealth', shorten >=0pt, auto, thick, outline/.style={draw=black, thick},fnode/.style={circle, inner sep=0pt, minimum size=4pt, draw=black, fill}]
    \draw[outline] (0,0) circle (0.85cm)
      node[fnode] (u1) at +({360/3 * (3 - 1)}:15pt) {}
      node[fnode] (u2) at +({360/3 * (3 - 2)}:15pt) {}
      node[fnode] (u3) at +({360/3 * (3 - 3)}:15pt) {};        
    \path[every node/.style={font=\sffamily\small}]
      (u1) edge node {} (u2)
           edge node {} (u3)
      (u2) edge node {} (u1)
           edge node {} (u3)
      (u3) edge node {} (u1)
           edge node {} (u2);

    \draw[outline] (0.75,2) circle (0.85cm)
      node[fnode] (v1) at +({360/3 * (3 - 1)}:15pt) {}
      node[fnode] (v2) at +({360/3 * (3 - 2)}:15pt) {}
      node[fnode] (v3) at +({360/3 * (3 - 3)}:15pt) {};        
    \path[every node/.style={font=\sffamily\small}]
      (v1) edge node {} (v2)
           edge node {} (v3)
      (v2) edge node {} (v1)
           edge node {} (v3)
      (v3) edge node {} (v1)
           edge node {} (v2);

    \path[every node/.style={font=\sffamily\small}]
      (u3) edge node {} (v1)
      (u2) edge node {} (v1)
      (u2) edge node {} (v2);      

    \draw (4,0) node[fnode] (w1) {} [below right] node {$w_1$};
    \draw (5,2) node[fnode] (w2) {} [above left] node {$w_2$};
    \path[every node/.style={font=\sffamily\small}]
      (w1) edge node {} (w2);
      
    \node[scale=2] at (2.7,0.8) {${\Longleftarrow}$};
    \node at (2.8,1.25) {$\phi^{-1}$};
    
    \node at (1.5,-0.55) {$\phi^{-1}(w_1)$};
    \node at (2.3,2.4) {$\phi^{-1}(w_2)$};    
\end{tikzpicture}
\end{center}
\caption{Demonstration of Lemma~\ref{lem:Homomorphism_bundle}. Part (i) states that inside each bundle $\phi^{-1}(w_i)$ we have a clique and part (ii) states that if $w_2$ is an outgoing neighbour of $w_1$ in $H$ then all of the vertices  in $\phi^{-1}(w_1)$ of $G$ are connected to all of the vertices in $\phi^{-1}(w_2)$ of $G$ (note that all of the edges from $\phi^{-1}(w_1)$ to $\phi^{-1}(w_2)$ are not shown in the figure).}\label{fig:Homomorphism_bundle}
\end{figure}
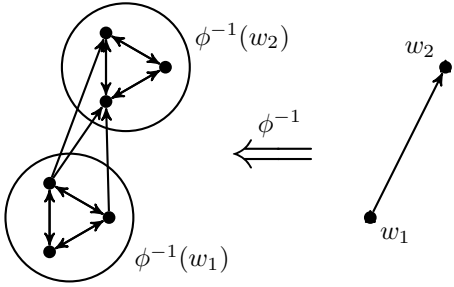

\begin{definition}
For every finite set $\set{X}$ and positive integer $m$, a function
$f:\set{X}^m\mapsto X$ is called \emph{coordinate-wise one-to-one} 
if by setting the values for every $m-1$ variables of $f$, it is 
a one-to-one function of the remaining variable, \ie, for every $j\in[1:m]$ and any choice of $a_1,\ldots,a_{j-1},a_{j+1},\ldots,a_m\in\set{X}$, the function
$f(a_1,\ldots,a_{j-1},x,a_{j+1},\ldots,a_m):\set{X}\mapsto\set{X}$ is
one-to-one.
\end{definition}

\begin{lemma}\label{lem:ExistenceCoordinate-wise_one-to-one}
For every finite set $\set{X}$ and $m\in\mbb{N}$, there exists a coordinate-wise one-to-one function.
\end{lemma}


\begin{proof}[Proof of Theorem~\ref{thm:MainResult_indRate_Homomorphism}]
Suppose that $V(H)=\{w_1,\ldots,w_n\}$ where $n=|V(H)|$. Let $\phi:\compl{G}\mapsto\compl{H}$
be a homomorphism. As stated in Lemma~\ref{lem:Homomorphism_bundle}, the
vertex set of $G$ can be partitioned into $n$ cliques of the form 
$\phi^{-1}(w_i)$. So, we can list the vertices of $G$ as 
$V(G)=\{v_{1,1},\ldots,v_{1,k_1},\ldots,v_{n,1},\ldots,v_{n,k_n}\}$ such that
$\phi^{-1}(w_i)=\{v_{i,1},\ldots,v_{i,k_i}\}$ and $k_i=|\phi^{-1}(w_i)|$.
Note that $m=|V(G)|=\sum_{i=1}^n k_i$.

Let $\ell=\ind_{\set{X}}(H)$ and $\Phi^H(x_1,\ldots,x_n):\set{X}^n\mapsto \set{Y}^\ell$ (in addition to a set of decoders $\{\Psi^H_i\}$) be an optimal valid index code for $H$ over the source alphabet $\set{X}$ (and the message alphabet $\set{Y}$) where $x_i$ is the variable associated to the node $w_i$. 

Validity of the index code implies that for every node $w_i\in H$, there exists
a decoding function $\Psi^H_i:\set{Y}^\ell\times \set{X}^{|N_H^+(w_i)|} \mapsto \set{X}$ such that 
$\Psi^H_i(\Phi^H(x_1,\ldots,x_n),x_{N_H^+(w_i)})=x_i$
for every choice of $(x_1,\ldots,x_n)\in\set{X}^n$.

Finally, we construct a valid index code for $G$ over the same alphabet sets and
the same transmission length $\ell$; thus it results in an index code for $G$
with the same broadcast rate. For an explicit construction see \cite{EbrahimiJafari-TechReport13}.
\end{proof}

As a result of Theorem~\ref{thm:MainResult_indRate_Homomorphism} we have the following corollary.
\begin{corollary}\label{cor:MainResult_indRate_Homomorphism}
Consider two instances of the index coding problems over the digraphs $G$ and
$H$. If $G\preccurlyeq H$ then we have
\begin{enumerate}
\item for general multi-letter index codes: $\ind(G) \le \ind(H)$,
\item for linear vector index codes: $\vlind_{q^t}(G) \le \vlind_{q^t}(H)$,
\item for linear scalar index codes: $\lind_{q}(G) \le \lind_{q}(H)$.
\end{enumerate}
\end{corollary}

\section{An Equivalent Formulation for Linear Scalar Index Coding Problem}\label{sec:EquivFormLinScalBinarIndexCoding}
Let $\set{G}^{q}_k$ be the set of all the finite digraphs $G$ for which  $\lind_q(G)\le k$. It is obvious to see that $\set{G}^{q}_k$ is an infinite family of digraphs. However, in this section, we will show that $\set{G}^{q}_k$ has a maximal member with respect to the pre order ``$\preccurlyeq$''. We give an explicit construction for a maximal element of $\set{G}^{q}_k$ which we call it $H^q_k$.

In fact, we show that $\forall G\in\set{G}^q_k$, $G\preccurlyeq H^q_k$. On the other hand, by Corollary~\ref{cor:MainResult_indRate_Homomorphism}, Part~	3, we know that if $G\preccurlyeq H^q_k$ then $\lind_q(G) \le \lind_q(H^q_k) \le k$. Thus, we can conclude the following theorem.

\begin{theorem}\label{thm:IndxCdPrblm_equiv_GrphHomomorphism}
For every positive integer $k$ and a prime power $q$, there exists a graph $H^q_k$ with $\frac{q^k-1}{q-1} q^{k-1}$ nodes such that for every graph $G$, $\lind_q(G)\le k$ if and only if $G\preccurlyeq H^q_k$ or equivalently, there exists a homomorphism from $\compl{G}$ to $\compl{H^q_k}$.
\end{theorem}



For the sake of simplicity, we prove the theorem for $q=2$. For general $q$, a construction for $H^q_k$ as well as a proof of the above theorem is presented in \cite{EbrahimiJafari-TechReport13}.

We start by presenting a construction for $H^2_k$. Consider a $k\times (2^k-1)$ binary matrix $B$ whose rows are labelled by numbers $1,2,\ldots, k$ and whose columns are labelled by non-empty subsets of $[1:k]$. For every $\varnothing\neq J \subseteq [1:k]$, the $J$-th column of $B$ is the indicator vector of the set $J$, \ie, the $(i,J)$-th entry of $B$ is $1$ iff $i\in J$.

Let $A$ be a $(2^k-1)\times (2^k-1)$ binary matrix whose rows and columns are indexed by non-empty subsets of $[1:k]$ and the $I$-th row of $A$ is equal to the binary summation (xor) of the rows of $B$ corresponding to the elements of $I$. Notice that $A_{(I,J)}=1$ iff $|I\cap J|$ is an odd number. Figure~\ref{fig:ExamplMatrix_A_for_K=2} shows an example of $A$ for $k=2$.

\begin{figure}
\begin{equation*}
A = \kbordermatrix{
    & \{1\} & \{2\} & \{1,2\} \\
    \{1\} & 1 & 0 & 1  \\
    \{2\} & 0 & 1 & 1  \\
    \{1,2\} & 1 & 1 & 0 
}
\end{equation*}
\caption{An example of the matrix $A_{3\times 3}$ for $k=2$.}
\label{fig:ExamplMatrix_A_for_K=2}
\vspace{-0.25cm}
\end{figure}

Now, define the digraph $H^2_k$ as follows. The set of vertices of $H^2_k$ is the set of pairs $(I,J)$ where $\varnothing\neq I,J \subseteq [1:k]$ and $A_{(I,J)}=1$. 
We denote the vertex of $H^2_k$ associated with $(I,J)$ by $v_{(I,J)}$. Equivalently, the vertices of $H^2_k$ are $v_{(I,J)}$ where $\varnothing\neq I,J\subseteq [1:k]$ and $|I\cap J|$ is an odd number.
The edges of $H^2_k$ are of the form $(v_{(I,J)},v_{(I',J')})$ such that $A_{(I,J')}=1$.
A simple way to visualize the digraph $H^2_k$ is the following. The vertex set of $H^2_k$ is the $1$'s of the matrix $A$ and there exists an edge from one vertex to another one if and only if the entry that is in the same row as the first vertex and in the same column of the second vertex is also equal to $1$. In particular, all the $1$'s that are in the same row (column) form a clique. An example of $H^2_k$ for $k=2$ is depicted in Figure~\ref{fig:Graph_Hk_for_K=2}. In this example, $(v_{(\{1\},\{1\})} , v_{(\{2\},\{1,2\})}) \in E(H_2)$ because of the entry $(\{1\},\{1,2\})$.

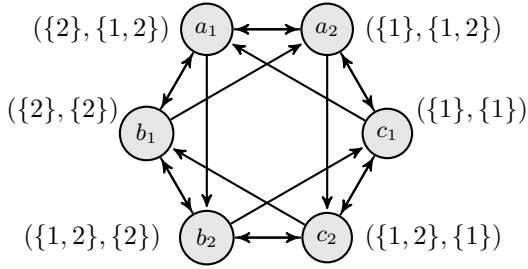
\begin{figure}
\begin{center}
\begin{tikzpicture}[->,>=stealth', shorten >=1pt, auto, node distance=2cm,
  thick, main node/.style={circle,fill=gray!20, draw, minimum size=6mm, font=\small}, scale=0.8]

  \draw ({360/6 * (3 - 1)}:2cm) node[main node] (1) {$a_1$} [anchor=east] node {$(\{2\},\{1,2\})$~~~~};  
  \draw ({360/6 * (2 - 1)}:2cm) node[main node] (2) {$a_2$} [anchor=west] node {~~~$(\{1\},\{1,2\})$};
  \draw ({360/6 * (4 - 1)}:2cm) node[main node] (3) {$b_1$} [anchor=south east] node {$(\{2\},\{2\})$~~~};
  \draw ({360/6 * (5 - 1)}:2cm) node[main node] (4) {$b_2$} [anchor=east] node {$(\{1,2\},\{2\})$~~~~~};
  \draw ({360/6 * (1 - 1)}:2cm) node[main node] (5) {$c_1$} [anchor=south west] node {~~$(\{1\},\{1\})$};  
  \draw ({360/6 * (6 - 1)}:2cm) node[main node] (6) {$c_2$} [anchor=west] node {~~~$(\{1,2\},\{1\})$};  
  
  \path[every node/.style={font=\sffamily\small}]
    (1) edge  node {} (2)
        edge  node {} (3)
        edge node {} (4)
    (2) edge  node {} (1)
        edge  node {} (5)
        edge node {} (6)
    (3) edge  node {} (1)
        edge  node {} (4)
        edge node {} (2)
    (4) edge  node {} (3)
        edge  node {} (6)
        edge node {} (5)
    (5) edge  node {} (2)
        edge  node {} (6)
        edge node {} (1)
    (6) edge  node {} (4)
        edge  node {} (5)
        edge node {} (3);
\end{tikzpicture}
\end{center}
\caption{The digraph $H^2_2$ consists of $6$ vertices. 
} 
\label{fig:Graph_Hk_for_K=2}
\vspace{-0.4cm}
\end{figure}

The next lemma explains the role of the family of digraph $H^2_k$'s in studying the scalar index coding problem.
\begin{lemma}
If $\lind_2(G)\le k$ then $G\preccurlyeq H^2_k$, \ie, there exists a homomorphism $\phi$ from $\compl{G}$ to $\compl{H^2_k}$.
\end{lemma}
\begin{proof}
The proof of this lemma is by constructing $\phi$. Suppose that $\lind_2(G)\le k$. Therefore, there exists an index coding scheme that transmits $k$ binary messages. 

Suppose that the $j$-th transmitted message is $y_j=\sum_{i\in M_j} x_i$ for the sets $M_1,\ldots,M_k\subseteq [1:m]$ where $x_i$ is the variable associated to the node $v_i\in V(G)$.
For every subset $\varnothing\neq J\subseteq [1:k]$ define 
$C_J \triangleq \bigcap_{j\in J} M_j \setminus \bigcup_{l\notin J} M_l$.

In other words, $C_J$ consists of all the indices that belong to all of $M_j$'s with $j\in J$ but no other $M_l$. From elementary set theory, it is easy to observe that $C_J$'s are pairwise disjoint; (see Figure~\ref{fig:Demonstration_of_set_CJ}). Moreover, note that the union of $C_J$'s is the whole set of $[1:m]$. This is due to the fact that all $M_i\subseteq [1:m]$ and therefore $C_J$'s are also subset of $[1:m]$ and if some element in $[1:m]$ is missing in all the $C_J$'s, it is also missing in all the $M_i$'s. That is, there exists a vertex $v_i\in V(G)$ such that its corresponding variable does not appear in any $y_j$. Equivalently, in none of the transmitted messages
the variable $x_i$ contributes. But this is a contradiction since the vertex $v_i$ cannot recover its demand only from its side information. So, $C_J$'s are $2^k-1$ disjoint subsets of $[1:m]$ which cover the whole set $[1:m]$. In fact, each $C_J$ consists of some indices such that for every message $y_j$, either all the variables of the form $x_i$, $i\in C_J$ appear, or none of them appear in $y_j$. Hence, each message $y_j$ can be written as a summation of $x_{[C_J]}$ where $x_{[C_J]}\triangleq \sum_{i\in C_J} x_i$. Therefore, 
$y_j = \sum_{J\ni j} x_{[C_J]}$.

\begin{figure}
\begin{center}
\begin{tikzpicture}[filled/.style={ draw=black, thick}, outline/.style={draw=black, thick},
    Mnode/.style={font=\small}, scale=0.9]

    \draw[outline] (0,0) circle (1cm);
    \draw[outline] (0:1cm) circle (1cm);
    \draw (-0.5,0) node[Mnode] {$C_{\{1\}}$};
    \draw (+1.5,0) node[Mnode] {$C_{\{2\}}$};
    \draw (+0.5,0) node[Mnode] {$C_{\{1,2\}}$};
    \draw (-1.2,-0.5) node[Mnode]  {$M_1$};    
    \draw (+2.3,-0.5) node[Mnode]  {$M_2$};        
\end{tikzpicture}
\end{center}
\caption{The relation between sets $M_j$'s and sets $C_J's$.}
\label{fig:Demonstration_of_set_CJ}
\vspace{-0.45cm}
\end{figure}
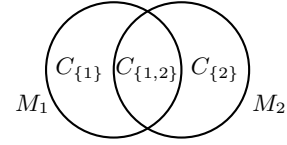

Next, we will define another partition of $[1:m]$ as follows. By definition of linear index coding, it is guaranteed that for every vertex $v_i$, the side information of $v_i$ and some subset of messages, transmitted by the source, will be enough for $v_i$ to recover $x_i$. The set of indices of every such subset of messages is called a \emph{sufficient family} for that particular receiver. Notice that for a particular receiver $v_i$, \emph{minimal} sufficient families\footnote{Here by a minimal sufficient family we refer to a sufficient family where none of its subsets is a sufficient family.} are not necessarily unique. For example, it is possible that a receiver can reconstruct its bit using the first two messages and its side information, or from the third message and its side information. However, there exists at least one minimal sufficient family of the messages. Let $\gamma:[1:m]\mapsto \mc{P}^\star([1:k])$ be a function that to every $i\in [1:m]$, $\gamma(i)$ assigns a minimal sufficient family for $v_i$. It is also easy to observe that if the set $\{y_j|j\in J\}$ is a minimal sufficient family of the messages for $v_i$ then $v_i$ is able to reconstruct $x_i$ using its side information and $\sum_{j\in J} y_j$.

For every $\phi\neq I\subseteq [1:k]$, define $D_I$ to be $\gamma^{-1}(I)$. Since $\gamma$ is a function, $D_I$'s for different $I$'s will partition the set $[1:m]$. Notice that by the definition of $\gamma$, if $i\in D_I$ then $\{y_j:j\in I\}$ is a minimal sufficient family for $v_i$. Therefore, $v_i$ can retrieve $x_i$ using its side information and $\sum_{j\in I} y_j$. That is to say that in the summation $\sum_{j\in I} y_j$, the variable $x_i$ appears and also, if another variable $x_{i'}$ appears, then $v_i$ knows $x_{i'}$ as its side information, \ie, $(v_i,v_{i'})\in E(G)$.

At this point, we are able to define a homomorphism from $\compl{G}$ to $\compl{H^2_k}$. In fact, we define a function $\phi:V(G)\mapsto V(H^2_k)$ and show that if $(v_i, v_{i'})\notin E(G)$ then $(\phi(v_i), \phi(v_{i'}))\notin E(H^2_k)$.

Since $C_J$'s and also $D_I$'s both partition the set $[1:m]$, for every $v_i\in V(G)$, there exists a unique pair of $(I,J)$, $\varnothing\neq I,J\subseteq [1:k]$ such that $i\in D_I$, $i\in C_J$. Define $\phi(v_i)=v_{(I,J)}$. In order to complete the proof, we need to show that $\phi$ is a well-defined function, \ie, $\phi(v)$ is a vertex of $H^2_k$ and also $(v_i, v_{i'})\notin E(G) \Rightarrow (\phi(v_i), \phi(v_{i'}))\notin E(H^2_k)$.

\begin{lemma}\label{lem:phi_is_well_defined}
The mapping $\phi$ is a well-defined function from $V(G)$ to $V(H^2_k)$, \ie, if $v_i\in V(G)$ then $A_{(I,J)}=1$ in which $I,J$ are such that $i\in D_I$ and $i\in C_J$.
\end{lemma}

\begin{lemma}\label{lem:no_edge_in_G_no_edge_in_H}
If $v_i,v_{i'}\in V(G)$ and $(v_i,v_{i'})\notin E(G)$ then $\left(\phi(v_i),\phi(v_{i'})\right) \notin E(H^2_k)$.
\end{lemma}
\end{proof}

So far, we have proved that if $\lind_2(G)\le k$ then $G\preccurlyeq H^2_k$. Conversely, if $G\preccurlyeq H^2_k$ then by Corollary~\ref{cor:MainResult_indRate_Homomorphism} we have $\lind_2(G)\le \lind_2(H^2_k)$. Therefore, the following lemma will finalize the proof of Theorem~\ref{thm:IndxCdPrblm_equiv_GrphHomomorphism}.

\begin{lemma}
For every positive integer $k$, $\lind_2(H^2_k) \le k$.
\end{lemma}


\section{Application}\label{sec:Application}
In this section, we will demonstrate several applications of the theorems stated in the previous sections.

\subsection{Upper Bounds}
Here, we will show that some of the earlier upper bounds are only special cases of Corollary~\ref{cor:MainResult_indRate_Homomorphism}.

\begin{example}
One of the earliest upper bounds on the $\lind_2(G)$ is $\chi(\compl{G})$ where $\chi(\cdot)$ is the chromatic number of a graph (\eg, see \cite{AlonLubetStavWeinHassid-Focs08}). Notice that in our framework, this result is an immediate consequence of Corollary~\ref{cor:MainResult_indRate_Homomorphism}, Part~3. That is, if $\chi(\compl{G})=r$ then $\compl{G}\rightarrow K_r$ where $K_r$ is a complete graph with $r$ vertices. Thus, $\lind_2(G)\le \lind_2(\compl{K_r})=r=\chi(\compl{G})$.
\end{example}

\begin{example}
In \cite{BlaKleLub-CoRR10}, it is shown that $\lind(G)\le \chi_f(\compl{G})$ where $\chi_f(\cdot)$ is the \emph{fractional chromatic number} of a graph. See \cite{EbrahimiJafari-TechReport13} for the proof using our framework. 
\end{example}

The crucial observation is that the parameters $\chi(\compl{G})$ and $\chi_f(\compl{G})$ can be defined using existence of homomorphisms from $\compl{G}$ to the family of complete graphs and Kneser graphs, respectively. (See \cite{EbrahimiJafari-TechReport13} for the definition of Kneser graphs).

\subsection{Lower Bounds}
By using Theorem~\ref{thm:IndxCdPrblm_equiv_GrphHomomorphism}, the following result can be proved.

\begin{lemma}\label{lem:GenLowerBoundUsingMaximalElement}
Suppose that $h$ is an increasing function on \mbox{$(\set{G},\preccurlyeq)$} and $r$ is an upper bound on $h(H_k)$. For every digraph $G$, if $h(G)>r$ then $\lind_2(G) > k$.
\end{lemma}
\begin{proof}
If $\lind_2(G)\le k$ then by Theorem~\ref{thm:IndxCdPrblm_equiv_GrphHomomorphism}, $G\preccurlyeq H_k$ and therefore $h(G)\le h(H_k)\le r$ which is a contradiction.
\end{proof}

Lemma~\ref{lem:GenLowerBoundUsingMaximalElement} is a powerful tool to find lower bounds on the index coding problem. Actually for every increasing function $h$ on $(\set{G},\preccurlyeq)$ we have one lower bound on the index coding problem. In the next theorem, we provide a lower bound on $\lind_q(G)$ in terms of the chromatic number of $\compl{G}$.

\begin{theorem}
For every digraph $G$, $\lind_q(G)\ge \log_q(\chi(\compl{G}))$.
\end{theorem}
\begin{proof}
The function $h(G)=\chi(\compl{G})$ is an increasing function on $(\set{G},\preccurlyeq)$. Suppose that $\lind_q(G)=k$. Therefore,  $\chi(\compl{G}) \leq \chi(\compl{H^q_k}) \leq q^k = q^{\lind_q(G)}$. The first inequality is implied by the previous Lemma. For a proof of the second inequality see \cite{EbrahimiJafari-TechReport13}. 
\end{proof}


\subsection{Index Codes and Change of Field Size}

Existence of a certain index code for a given graph over a fixed finite field is equivalent to the existence of linear combinations of the source messages over the ground field with certain Algebraic / Combinatorial constraints. If the ground field is changed, there is no natural way of updating the index code over the new field. In other words, if for a fixed graph $G$, an index code over a finite field $\mathbb{F}_{q_1}$ is given, there is no natural way to construct some index code for the same graph but over a different field $\mathbb{F}_{q_2}$. In fact, in \cite{LubetStav-IT09}, it has been shown that for every pair of finite fields $\mathbb{F}_p$ and $\mathbb{F}_q$ of different characteristics and for every $0 < \epsilon < 0.5$, there exists a graph $G$ with $n$ vertices such that $\lind_{p}(G) < n^{\epsilon} , \lind_q(G)>n^{1-\epsilon}$.  

Here we use the results of Theorem 1 and Theorem 2 to show that if $\lind_p(G)$ is less than a fixed number then by changing the field size, the corresponding linear indices can at most differ by a factor that depends only on the field sizes and is independent from the size of the graph. More precisely the following result holds:

\begin{theorem}
Let $G$ be a graph and $q_1,q_2$ are two different prime powers. Then, $\lind_{q_2}(G)\leq \lind_{q_2} (H^{q_1}_{\lind_{q_1}(G)})$.
\end{theorem}  
\begin{proof}
Suppose that $\lind_{q_1}(G) =k$. By Theorem 2, $G \preccurlyeq H^{q_1}_k$ and then by Theorem 1, $\lind_{q_2}(G) \leq \lind_{q_2} (H^{q_1}_k)$.
\end{proof}



\bibliographystyle{IEEEtran}
\bibliography{index_coding}

\end{document}